\documentclass[envcountsame]{llncs}
\usepackage{ifpdf}
\ifpdf
\usepackage{pdfsync}
\fi
\usepackage{amsmath}
\usepackage{amssymb}
\usepackage{amsfonts}
\usepackage{mathdots}
\usepackage{mathrsfs}
\usepackage{mathtools}
\usepackage{enumerate}
\usepackage{gastex}
\usepackage{xcolor}
\usepackage{latexsym}
\usepackage{url}
\usepackage{comment}
\usepackage{color}
\usepackage{graphicx,wrapfig}
\usepackage{epstopdf}
\usepackage{algorithm}
\usepackage{algpseudocode}
\usepackage{verbatim}
\algnewcommand\algorithmicinput{\textbf{INPUT:}}
\algnewcommand\INPUT{\item[\algorithmicinput]}
\algnewcommand\algorithmicoutput{\textbf{OUTPUT:}}
\algnewcommand\OUTPUT{\item[\algorithmicoutput]}

\usepackage{hyperref}

\algtext*{EndWhile}% Remove "end while" text
\algtext*{EndIf}% Remove "end if" text
\algtext*{EndFunction}% Remove "end if" text
\algtext*{EndFor}% Remove "end for" text

\usepackage[utf8]{inputenc}
\usepackage{tikz}
    \usetikzlibrary{arrows, automata,  shapes.geometric, shapes.misc}

    \tikzstyle{block} = [
        rectangle, draw,
        text width=4em, text centered,
        minimum height=3em
    ]

    \tikzstyle{decision} = [
        diamond, draw, text width=4em, text centered, minimum height =3em
    ]
    \tikzstyle{line} = [draw, -latex']

\newcommand{\keywords}[1]{\par\addvspace\baselineskip
\noindent\keywordname\enspace\ignorespaces#1}

\newcommand{\A}{\mathcal A}

\newcommand{\C}{\mathsf{C}}

\newcommand{\F}{\mathcal F}
\newcommand{\FR}{\mathsf{F}}
\newcommand{\false}{\mathsf{false}}

\newcommand{\N}{\mathbb N}
\newcommand{\NR}{\mathsf{S}}

\newcommand{\R}{\mathsf{R}}

\newcommand{\true}{\mathsf{true}}

\floatname{algorithm}{Procedure}

\spnewtheorem{example}{Example}{\bfseries}{\rmfamily}

\begin{document}

\title{A Game of Attribute Decomposition for Software Architecture Design}

\author{Jiamou Liu  \ and \ Ziheng Wei}
\institute{School of Computer and Mathematical Sciences \\ Auckland University of Technology, New Zealand\\
\email{\{jiamou.liu, bys7090\}@aut.ac.nz}
}

\maketitle

\begin{abstract}
Attribute-driven software architecture design aims to provide decision support by taking into account the quality attributes of softwares. A central question in this process is: {\em What architecture design best fulfills the desirable software requirements?} \ To answer this question, a system designer  needs to make tradeoffs among several potentially conflicting quality attributes. Such decisions are normally ad-hoc and rely heavily on experiences. We propose a mathematical approach to tackle this problem. Game theory naturally provides the basic language:  Players represent requirements, and strategies involve setting up coalitions among the players. In this way we propose a novel model, called {\em decomposition game (DG)}, for attribute-driven design. We present its solution concept based on the notion of cohesion and expansion-freedom and prove that a solution always exists. We then investigate the computational complexity of obtaining a solution. The game model and the algorithms may serve as a general framework for providing useful guidance for software architecture design. We present our results through running examples and a case study on a real-life software project.

\end{abstract}

\section{Introduction}%\vspace*{-0.4cm}
 Computational game theory studies the algorithmic nature of conflicting entities and establishes {\em equilibria}: A state of balance that minimises the negative effects among players. The field has attracted much attention in the recent 10-15 years due to applications in multi-agent systems, electronic markets and social networks \cite{PC2001,RTV2007,SK2008}. In this paper, we investigate the problem of software architecture design from a game theory perspective. In particular, we provide a novel model, called {\em decomposition game}, which captures interactions among software requirements and derives a software architecture through  equilibria.

The architecture of a software system lays out its basic composition.
For softwares become larger, quality attributes such as performance, reliability, usability and security, play an increasingly important role.
It has been a common belief that architecture design heavily influences the quality attributes such as performance, reliability, usability and security of a software system \cite{B07}. A major objective of architecture design is therefore the assurance of non-functional requirements through  compositional decisions. In other words, we need to answer the following question: {\em What architecture best fulfills the desirable software requirements?}
There is, however, usually no ``perfect'' architecture that fulfills every requirement. For example, performance and security are both  key non-functional requirements, which may demand fast response time to the users, and the application of a sophisticated encryption algorithm, respectively. These two requirements are in intrinsic conflict, as a strong focus of one will negatively impact the fulfilment of the other.
A main task of the software architect, therefore, is to balance such ``interactions'' among requirements, and decide on appropriate tradeoffs among such conflicting requirements.

While it is a common practice to decide on software architecture designs through the designers' experiences and intuition, formal approaches for architecture design are desirable as they facilitate standardisation  and automation of this process, providing rigorous guidelines, allowing automatic analysis and verifications \cite{GD2003}. Notable formal methods in software architecture include a large number of formal architecture description languages (ADL), which are useful tools in communicating and modeling architectures. However, as argued by \cite{WH2005}, industry adoptions of ADL are rare due to limitations in usability and formality. Other algorithmic methods for software architecture design include employing hierarchical clustering algorithms to decompose components based on their common attributes \cite{LXZ2007}, as well as quantifying tradeoffs between requirements \cite{AADHMH2011}.

In this paper, we propose to use computational game theory as a mathematical foundation  for conceptualising software architecture designs from requirements. Our motivation comes from the following two lines of research:
%\vspace*{-0.4cm}
%\begin{enumerate}[(1)]

\paragraph{(1). Attribute driven design (ADD)}: ADD is a systematic method for software architecture design. The method was invented by Bass, Klein and Bachmann in \cite{BLMF2002} and subsequently updated and improved through a sequence of works \cite{G01,WR2006}. The goal is to assist designers to analyse quality attribute tradeoffs and provide design suggestions and guidance. Inputs to ADD are functional and non-functional requirements, as well as design constraints; outputs to ADD are conceptual  architectures which outline  coarse-grained system compositions.   The method involves a sequence of well-defined steps that recursively decompose a system to components, subcomponents, and so on. These steps are not algorithmic: They are meant to be followed by system designers based on their experience and understanding of design principles. As mentioned by the authors in \cite{BLMF2002}, an ongoing effort is to investigate rigorous approaches in producing conceptual architectures from requirements, hence enabling automated design recommendation under the ADD framework. To this end, we initiate a game-theoretic study to formulate the interactions among software requirements so that a conceptual architecture can be obtained in an algorithmic way.

%\vspace*{-0.4cm}
\paragraph{(2). Coalition game theory}: A coalition game is one where players exercise collaborative strategies, and competition takes place among coalitions of players rather than individuals.  In ADD, we can imagine each requirement is ``handled'' by a player, whose goal is to set up a coalition with others to maximise the collective payoff. The set of coalitions then defines components in a system decomposition which entails a software architecture. This fits into the language of coalition games. However, the usual axioms in coalition games specify super-additivity and monotonicity, that is, the combination of two coalitions is always more beneficial than each separate coalition, and the payoff increases as a coalition grows in size. Such assumptions are not suitable in this context as combination of two conflicting requirements may result in a lower payoff. Hence a new game model is necessary to reflect the conflicting nature of requirements. In this respect, we propose that our model also enriches the theory of coalition games.
%\end{enumerate}
%\vspace*{-0.3cm}

\paragraph*{\bf Our contribution.} %In Section~\ref{sec:ADD},
We provide a formal framework which, following the ADD paradigm \cite{BLMF2002},  recursively decomposes a system into sub-systems; the final decomposition reveals design elements in a software architecture. The basis of the framework is an algorithmic realisation of ADD. A crucial task in this algorithmic realisation is {\em system decomposition}, which derives a rational decomposition of an attribute primitive. %an attribute primitive is defined as a relational structure,
%whose domain consists of functional and non-functional requirements,
%while design elements are abstracted into sets of requirements.% (Section~\ref{sec:ADD}).
%
%A crucial step in this algorithmic realisation is system decomposition, which derives a rational decomposition of an attribute primitive.
%In Section~\ref{sec:game},
%

We model system decomposition using a game, which we call {\em decomposition game}. The game takes into account  {\em interactions} between requirements, which express the positive (enhancement) or negative (canceling) effects they act on each other. A {\em solution concept} (equilibrium) defines  a rational decomposition, which is based on the notions of {\em cohesion} and {\em expansion-freedom}.

We demonstrate that any such game has a solution, and a solution may not be unique. % (Section~\ref{sec:game})
%In Section~\ref{sec:algorithm},
We also investigate algorithms that compute solutions for the decomposition game.
Finding cohesive coalitions with maximal payoff turns out to be  NP-hard (Thm.~\ref{thm:NPComplete}). Hence we propose a relaxed notion of {\em $k$-cohesion} for $k\geq 1$, and present a polynomial time algorithm for finding a $k$-cohesive solution of the game (Thm.~\ref{thm:algo}). To demonstrate the practical significance our the framework, we implement the framework and perform a case study on a real-world Cafeteria Ordering System.

\paragraph*{\bf Paper organisation.} Section~\ref{sec:ADD} introduces the formal ADD framework. Section~\ref{sec:game} discusses decomposition game and its solution concept. Section~\ref{sec:algorithm} presents algorithms for solving decomposition games. Section~\ref{sec:case} presents the case study. Section~\ref{sec:related} discusses related works and finally Section~\ref{sec:conclusion} concludes with future works.

%The decomposition game, as well as its solutions concepts, are the major conceptual contributions of the paper, while the definition of requirement interactions, and the algorithmic analysis amount to the main technical contributions.

%\vspace*{-0.6cm}
\section{Algorithmic Attribute Driven  Design (ADD) Process} \label{sec:ADD} ADD is a general framework for transforming software requirements into a {\em conceptual software architecture}. Pioneers of this approach introduced it through several well-formed, but informally-defined concepts and steps \cite{BLMF2002,WR2006}.
%However, this approach has been used largely in empirical architecture design, and has not been formally described.
A natural question arises whether it can be made more algorithmic, which provides unbiased, mathematically-grounded outputs. To answer this question, one would first need to translate the original informal descriptions to a mathematical language. \vspace*{-0.9cm}
\subsection{Software Requirements and Constraints}
%\vspace*{-0.2cm}
\paragraph*{\bf Functional requirements.} Functional requirements are specifications of what tasks the system  perform (e.g. ``the system must notify the user once a new email arrives'').%; they are usually represented by use cases.
A functional requirement does not stand alone; often, it acts with other functional requirements to express certain combined functionality (e.g. ``the user should log in before making a booking''). Thus, a functionality may depend on other functionalities. We use a partial ordering $(\FR, \prec)$ to denote the functional requirements where each $r\in \FR$ is a functional requirement, and $r_1\prec r_2$ denotes that $r_1$ depends on $r_2$. Note  that $\prec$ is a transitive relation.
%\vspace*{-0.4cm}
\paragraph*{\bf Non-functional requirements.} Non-functional requirements specify the desired quality attributes; ADD uses {\em general scenarios} and {\em scenarios} as their standard representations. A general scenario is a high-level description on what it means to achieve a non-functional requirement \cite{BLMF2002}.
%(e.g. ``security: deny access to unauthorized users'')
%In the ADD framework, each general scenario consists of a {\em stimulus} and a {\em response}.
For example, the general scenario ``\emph{A failure occurs and the system notifies the user; the system continues to perform in a degraded manner}'' refers to the availability attribute.
%; the stimulus is the failure while the response will be the notification to the user.
There has been an effort to document all common general scenarios; a rather full list is given in \cite{G01}.
Note that a general scenario is vaguely-phrased and is meant to serve as a template for more concrete ``instantiations'' of quality attributes. %and thus is vaguely-phrased. In ADD,
Such ``instantiations''  are called scenarios. More abstractly, we use a pair $(\NR, \approx)$ to denote the non-functional requirements where  $\NR$ is a set of scenarios and $\approx$ is an equivalence relation on $\NR$, denoting the {\em general scenario relation}: $q_1\approx q_2$ means that $q_1$ and $q_2$ instantiates the same general scenario.
%\vspace*{-0.3cm}
\paragraph*{\bf Design constraints.} Design constraints are factors that must be taken into account and enforce certain design outcomes. A design constraint
%is usually informally specified and
may affect both functional and non-functional requirements. More abstractly, we use a collection of sets $\C\subseteq 2^{\FR\cup \NR}$ to denote the set of design constraints, where each set $c\in \C$ is a design constraint.  Intuitively, if two requirements $r_1,r_2$ belong to the same $c\in \C$, then they are constrained by the same design constraint $c$.
%\vspace*{-0.3cm}
\paragraph*{\bf Derived Functionalities.} The enforcement of certain quality attributes may lead to additional functionalities. For example, to ensure availability, it may be necessary to add extra functionalities to detect failure and automatically bypass failed modules. Hence we introduce a {\em derivation relation} $\hookrightarrow\subseteq \NR \times \FR$ such that $r\hookrightarrow s$ means the functional requirement $s$ is derived from the scenario $r$.
%\vspace*{-0.4cm}

\subsection{Attribute Primitives}

\begin{figure} %\vspace*{-1.6cm}
\centering
\caption{\label{fig:requirement}\small Example~\ref{exp:requirement}: The requirements, constraints and their relations.}
\includegraphics[width=6cm]{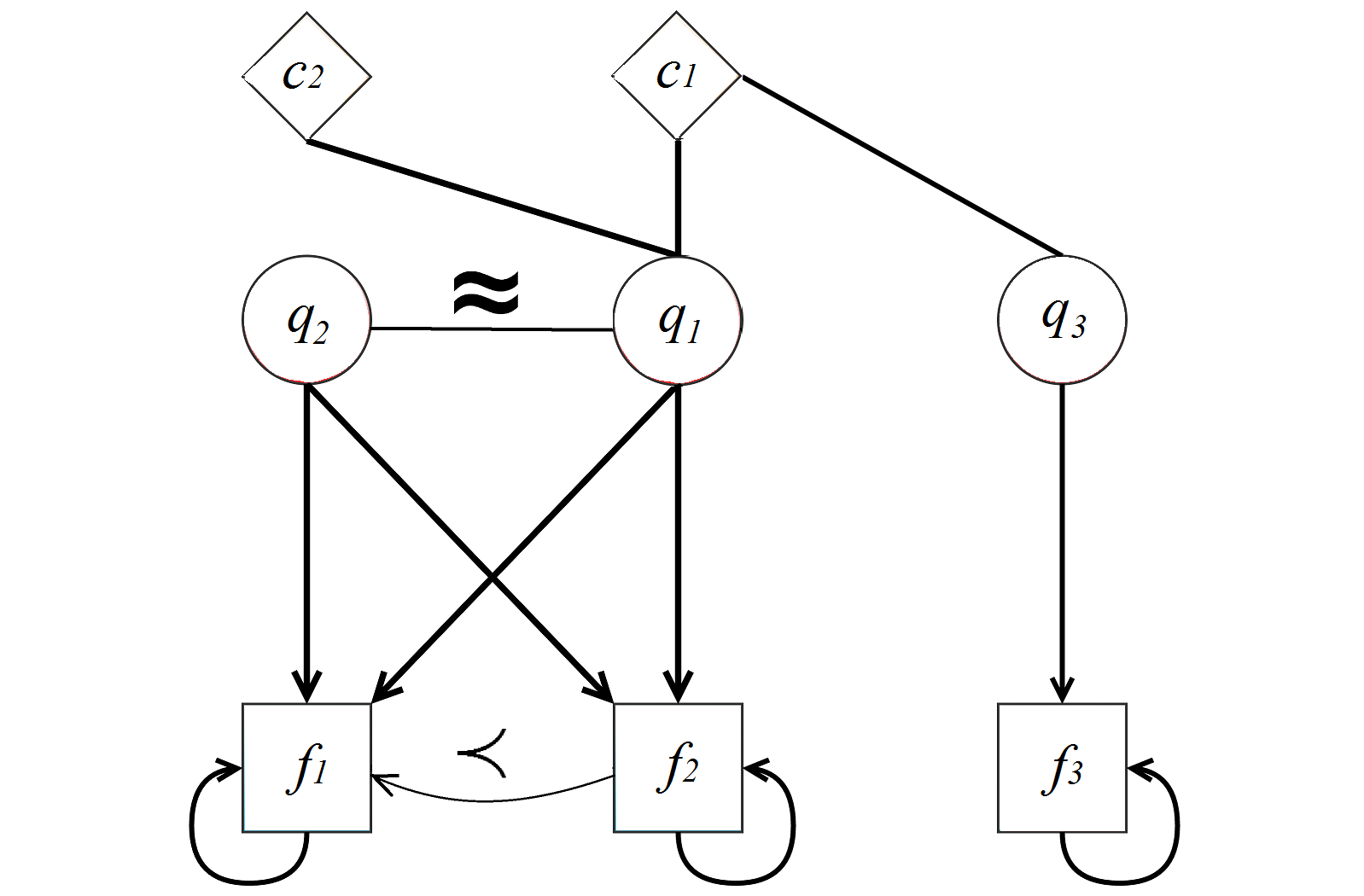}%\vspace*{-0.4cm}
\end{figure}%\vspace*{-0.2cm}
 The intentional outcome of ADD describes the
 {\em design elements}, i.e., subsystems, components or connectors.
 It is important to note that the goal of ADD is not the complete automation of the design process, but rather, to provide useful guidance. Thus, the conceptual view reveals only the organisational structure but not the concrete design.
  %for implementation. % (such as adopting an object-oriented paradigm).

%In ADD system modules are described by {\em attribute primitives}.
An attribute primitive is a set of design elements that collaboratively perform certain functionalities and meet one or more quality requirements; it is also the minimal combination with respect to these goals \cite{BLMF2002}. Examples of attribute primitives include data router, firewall,  virtual machine, interpreter and so on. ADD prescribes a list of attribute primitives together with descriptions of their properties and side effects (such as in \cite{G01}). Hence, ADD essentially can be viewed as assigning the right attribute primitives to the right requirement combinations.  Note also that an attribute primitive  may be broken down further. %This allows a recursive procedure for decompositions of primitive attributes.
 %\vspace*{-0.1cm}
 \begin{definition}[Attribute primitive]
 An {\em attribute primitive} is a tuple%\vspace*{-0.2cm}
 $$\A=(\FR, \NR, \C, \prec,\approx, \hookrightarrow)\vspace*{-0.2cm}$$ where $\FR$ is a set of functional requirements, $\NR$ is a set of scenarios, $\C\subseteq 2^{\FR\cup \NR}$ is a set of design constraints, $\prec$ is the dependency relation on $\FR$, $\approx$ is the general scenario relation of $\NR$,   and $\hookrightarrow\subseteq \NR\times \FR$ is a derivation relation.
 \end{definition}
 Let $\A=(\FR, \NR, \C, \prec,\approx,\hookrightarrow)$ be an attribute primitive. We also need the following definition: %\vspace*{-0.2cm}
 \begin{itemize}
 \item A {\em requirement} of $\A$ is an element in the set $\R\coloneqq\FR\cup \NR$.
 \item For $r\in \FR$, the {\em dependency set} of $r$ is the set $f(r)\coloneqq\{r'\in \FR\mid r\preceq r'\}$.
 \item For $r\in \NR$, the {\em general scenario} of $r$ is the set $g(r)\coloneqq\{r'\in \NR\mid r\approx r'\}$, i.e., the $\approx$-equivalence class of $r$.
 \item For $r\in \R$, the {\em constraints} of $r$ is the set $c(r)\coloneqq\{t\in \C\mid r\in t\}$.
 \item For $r\in \NR$, the {\em derived set} of $r$ is $d(r) \coloneqq \{s\in \FR\mid r\hookrightarrow s\}$, and for $s\in \FR$, let $d^{-1}(s) \coloneqq \{r\in \NR \mid r\hookrightarrow s\}$
 \end{itemize}
 %\end{definition}
 %
 %\vspace*{-0.2cm}
%Under this formalism, design elements are simply collections of requirements. \vspace*{-0.2cm}
 \begin{definition}[Design element]
 A {\em design element} of $\A$ is a subset $D\subseteq \R$.   %; if $D=\R$, then we say that $\A$ is {\em consistent} with $D$.
 An {\em decomposition} of $\A$ is a sequence of design elements
$ \vec{D}=(D_1, D_2,\ldots D_k)$
 where $k\geq 1$, $\bigcup_{1\leq i\leq k} D_k = \R$, and each $D_i\cap D_j=\varnothing$ for any $i\neq j$.
 \end{definition}\vspace*{-0.2cm}
 \begin{example}\label{exp:requirement}
Fig.~\ref{fig:requirement} shows an attribute primitive $\A=(\FR, \NR, \C, \prec,\approx,\hookrightarrow)$\vspace*{-0.2cm}
	\begin{itemize}
		\item  $\FR = \{f_1, f_2, f_3\}$ and $\NR = \{q_1, q_2, q_3\}$ are the requirements
		\item  $\C = \{c_1, c_2\}$ where $c_1=\{q_1,q_3\}, c_2=\{q_1\}$
		\item  $f_1 \prec f_2$, $q_1 \approx q_2$, $q_1 \hookrightarrow f_1$, $q_1 \hookrightarrow f_2$ ,$q_2 \hookrightarrow f_1$, $q_2 \hookrightarrow f_2$, $q_3 \hookrightarrow f_3$.
	\end{itemize}%\vspace*{-0.4cm}
%See Fig.~\ref{fig:requirement} for an illustration.
    \end{example}
%\vspace*{-1cm}
%\begin{figure}
%\begin{center}
% \includegraphics[width=5cm]{definition.png}
% \caption{\label{fig:requirement}\small Example~\ref{exp:requirement}: The requirements, constraints and their relations.}
%\end{center}\vspace*{-1.5cm}
%\end{figure}
%\vspace*{-0.4cm}
\subsection{The ADD Procedure}%\vspace*{-0.2cm}
Essentially ADD provides a means for system decomposition: The entire system is treated as an attribute primitive, which is the input. At each step, the procedure decomposes an attribute primitive $\A$ by identifying a decomposition $(D_1,D_2,\ldots,D_k)$. The process then maps each resulting design element $D_i$ to an attribute primitive $\A_i=({\FR}_i, {\NR}_i, \C_i, \prec_i, \approx_i, \hookrightarrow_i)$, which contains all elements in $D_i$ and  may require some further requirements and constraints. Hence we require that $D_i\subseteq {\FR}_i \cup {\NR}_i$ and $\prec_i$, $\approx_i$, $\C_i$, $\hookrightarrow_i$  are consistent with $\prec$, $\approx$, $\C$ and $\hookrightarrow$ on $D_i$, resp.; in this case we say that $\A_i$ is {\em consistent} with $D_i$. Thus the attribute primitive $\A$ is decomposed into $k$ attribute primitives $\A_1,\A_2,\ldots,\A_k$. On each $\A_i$ where $1\leq i\leq k$, the designer may choose to either terminate the process, or start a new step recursively to further decompose $\A_i$. See Procedure~\ref{prc:ADD}.

%\vspace*{-0.5cm}
\begin{algorithm}
\caption{$\mathsf{ADD}(\A)$ (General Plan)} \label{prc:ADD}
\begin{algorithmic}[1]
\State $(D_1,D_2,\ldots,D_k)\gets \mathsf{Decompose}(\A)$  // compute a rational decomposition of $\A$
\For{$1\leq i\leq k$}
    \State $\A_i \leftarrow $ an primitive attribute consistent with $D_i$
    \If{$\A_i$ needs further decomposition}
        \State $\mathsf{ADD}(\A_i)$
    \EndIf
\EndFor
\end{algorithmic}
\end{algorithm}%\vspace*{-0.7cm}
 %Procedure~\ref{prc:ADD} formally describes the ADD process.
We point out that the ADD procedure, as presented by its original proponents, involves numerous additional stages other than the ones described above \cite{WR2006}. The reason we choose this over-simplified description is that we believe these are the steps that could be rigorously presented, and they abstractly  capture in a way most of the steps mentioned in the original informal description. % (e.g. the  the operation $\mathsf{Decompose}(\A)$ captures the step  ``\emph{choosing attribute primitives to satisfy architectural drivers}'' in \cite{WR2006}; the operation at  Line 3 captures ``\emph{verify and refine use cases and scenarios and make them constraints for children design elements}'').

 The $\mathsf{Decompose}(\A)$ operation produces a rational decomposition $(D_1,\ldots,D_k)$ of the input attribute primitive $\A$ that satisfies the requirements of $\A$.
We also note that $\mathsf{Decompose}(\A)$ amounts to a crucial step in the ADD process, as the decomposition determines to a large extend how well the quality attributes are met. This step is also a challenging one as interactions among quality attributes create potential conflicts. Thus, in the next section, we define a game model which allows us to automate the $\mathsf{Decompose}(\A)$ operation.

%\vspace*{-0.3cm}
\section{Decomposition Games}\label{sec:game}
%\vspace*{-0.2cm}
\subsection{Requirement Relevance}\label{subsec:interaction}
The $\mathsf{Decompose}(\A)$ procedure  looks for a rational decomposition that meets the requirements in $\A$ as much as possible.
Let $\A=(\FR, \NR, \C, \prec,\approx,\hookrightarrow)$ be an attribute primitive. Relevance between requirements are determined by the relations $\prec,\approx,\hookrightarrow$ and the constraint set $\C$. In the following the {\em Jaccard index} $J(S_1,S_2)$ measures the similarity between two sets $S_1,S_2$ with
\[
    J(S_1,S_2) = \frac{|S_1\cap S_2|}{|S_1\cup S_2|}. %\vspace*{-0.2cm}
\]
Intuitively, the relevance of a requirement $r$ to other requirements is influenced by the ``links'' between $r$ and the functional, the non-functional requirements, as well as design constraints.
%\vspace*{-0.2cm}
\begin{definition}[Relevance] Two requirements $r_1,r_2\in \R$ are {\em relevant} if %\vspace*{-0.2cm}
\begin{itemize}
\item  $r_1,r_2\in \FR$, and either $d^{-1}(r_1)\cap d^{-1}(r_2)\neq \varnothing$  (derived from some common scenario), or $f(r_1)\cap f(r_2)\neq \varnothing$ (relevant through dependency), or $c(r_1)\cap c(r_2)\neq \varnothing$ (share some common design constraints).
\item  $r_1,r_2\in \NR$, and either $r_1\approx r_2$ (instantiate the same general scenario), or $d(r_1)\cap d(r_2)\neq \varnothing$ (jointly derives some functionality) or $c(r_1)\cap c(r_2)\neq \varnothing$.
\item $r_1\in \FR$, $r_2\in \NR$, and  either $f(r_1) \cap d(r_2) \neq \varnothing$ ($r_1$ depends on a requirement that is derived from $r_2$), or $c(r_1)\cap c(r_2) \neq \varnothing$.
\end{itemize}
\end{definition}
If two requirements are relevant, their relevance  depends on overlaps between their derived sets, dependency sets and constraints. If two requirements are not relevant, then we regard them as having a negative relevance $\lambda<0$, which represents a ``penalty'' one pays when two irrelevant requirements get in the same design element.
\begin{definition}
 We define the {\em relevance index} $\sigma(r_1,r_2)$ of  $r_1\neq r_2\in \R$ as follows:%\vspace*{-0.2cm}
\begin{enumerate}
\item if two functional requirements $r_1,r_2\in \FR$ are relevant, then\vspace*{-0.2cm}
\[
    \sigma(r_1,r_2) = \alpha J(d^{-1}(r_1),d^{-1}(r_2)) + \beta J(f(r_1),f(r_2)) + \gamma J(c(r_1),c(r_2));\vspace*{-0.2cm}
\]
\item if two scenarios $r_1,r_2\in \NR$ are relevant, then\vspace*{-0.2cm}
\[
    \sigma(r_1,r_2) = \beta J(d(r_1),d(r_2)) + \gamma J(c(r_1),c(r_2));\vspace*{-0.2cm}
\]
\item If $r_1\in \FR$ and $r_2\in \NR$ are relevant, then\vspace*{-0.2cm}
\[
    \sigma(r_1,r_2) = \sigma(r_2,r_1)= \beta J(f(r_1), d(r_2)) + \gamma J(c(r_1),c(r_2));\vspace*{-0.2cm}
\]
\item otherwise, $\sigma(r_1,r_2)=\lambda$
\end{enumerate}\vspace*{-0.1cm}
The constants $\alpha,\beta,\gamma$ are positive real numbers, that represent weights on the overlaps in $d_1,d_2$'s generated sets,
dependency sets and constraints, respectively. We require $\alpha\!+\!\beta\!+\!\gamma\!=\!1$.
\end{definition}
For simplicity, we do not include these constants in expressing the function $\sigma$, and all subsequent notions that depend on $\sigma$ (thus saving us from writing ``$\sigma(r_1,r_2,\alpha,\beta,\gamma,\lambda)$'').

%\vspace*{-0.2cm}
\begin{example}\label{exp:relevance} Continue from $\A$ in Example~\ref{exp:requirement}.
To emphasise the  non-functional requirements we give a larger weight to $\alpha$, setting $\alpha = 0.5$, $\beta=0.4$, $\gamma=0.1$. We also set $\lambda=-0.5$. Then
$\sigma(r_1,r_2)\!=\!0.4\!\times\!\frac{2}{2}\!=\!0.4$
for any $(r_1,r_2)\!\in\!\{(q_1,q_2),(q_3,f_3)\}\!\cup\!(\{q_1,q_2\}\!\times\!\{f_1,f_2\})$; $\sigma(q_1, q_3)\!=\!0.1\!\times\!\frac{1}{2}\!=\!0.05$; $\sigma(f_1, f_2)\!=\!0.5 \times \frac{2}{2} + 0.4 \times \frac{2}{2} = 0.9$;
and relevance between any other pairs is $-0.5$. Fig.~\ref{fig:relevance}(a) illustrates the (positive) relevance in a weighted graph.
\end{example}
%\vspace*{-0.7cm}
\subsection{Decomposition Games} \label{subsec:game}%An important question is what constitutes a ``rational'' decomposition. For this we employ notions from coalitional game theory.
%\vspace*{-0.2cm}
We employ notions from coalition games to define what constitutes a {\em rational} decomposition.
In a coalition game,  players cooperate to form coalitions which achieve certain collective payoffs \cite{MO2008}.\vspace*{-0.1cm}
\begin{definition}[Coalition game]
A {\em coalition game} is a pair $(P,\nu)$ where
 $P$ is a finite set of players, and each subset $D\subseteq P$ is  a {\em coalition};
 $\nu: 2^P \to \mathbb{R}$ is a {\em payoff function} associating every  $D\subseteq P$ a real value $\nu(D)$ satisfying $\nu(\varnothing)=0$.\vspace*{-0.1cm}
\end{definition}
This provides the set up for decompositions: \  Imagine a coalition game consisting of $|\R|$ agents as players, where each agent is in charge of a different requirement. The players form coalitions which correspond to sets of requirements, i.e., design elements. The payoff function would associate with every coalition a numerical value, which is the payoff gained by each member of the coalition. Therefore, an equilibrium of the game amounts to a decomposition with the right balance among all requirements -- this would be regarded as a rational decomposition.

It remains to define the payoff function.  Naturally, the payoff of a coalition is determined by the {\em interactions} among its members. Take $r_1,r_2\in D$. If one of $r_1,r_2$ is a functional requirement, then their interaction is defined by their relevance index $\sigma(r_1,r_2)$, as higher relevance means a higher level of interaction. Suppose now both $r_1,r_2$ are scenarios (non-functional). Then the interaction becomes more complicated, as a quality attribute may enhance or defect another quality attribute. In \cite[Chapter 14]{WKJ2013}, the authors identified effects acting from one quality attribute to another, which is expressed by a {\em tradeoff matrix} $T$:%\vspace*{-0.2cm}
\begin{itemize}
\item $T$ has dimension $m\times m$ where $m$ is the number of general scenarios
\item  For $ i\neq j \in \{1,\ldots,m\}$, the $(i,j)$-entry $T_{i,j}\in \{-1,0,1\}$.
\end{itemize}%\vspace*{-0.2cm}
Let $g_1,g_2,\ldots,g_m$ be general scenarios. $T_{i,j}=1$ (resp. $T_{i,j}=-1$) means  $g_1$ has a positive (resp. negative) effect on $g_2$, $T_{i,j}=0$ means  no effect.
%;  $T(g_i,g_j)$ denotes $T_{i,j}$.
E.g., the tradeoff matrix  defined on six common quality attributes is:%\vspace*{-0.2cm}
{\[
\begin{array}{c|cccccc}
                  & \text{\scriptsize \ Perfo.\ }  & \text{\scriptsize Modif.\ } & \text{\scriptsize Secur.\ } & \text{\scriptsize Avail.\ } & \text{\scriptsize Testa.\ } & \text{\scriptsize Usabi.\ } \\ \hline
\text{\scriptsize Performance} &0 & -1 & 0 &  0 & 0 & -1 \\
\text{\scriptsize Modifiability} &-1 & 0 & 0 & 1 & 1 & 0 \\
\text{\scriptsize Security} &-1 & 0 & 0 & 1 & -1 & -1 \\
\text{\scriptsize Availability} &0 & 0 & 0 & 0 & 0 & 0 \\
\text{\scriptsize Testability} &0 & 1 & 1 & 1 & 0 & 1 \\
\text{\scriptsize Usability} &-1 & 0 & 0 & 0 & -1 & 0
\end{array}%\vspace*{-0.2cm}
\]}
Note that the matrix is not necessarily symmetric: The effect from $g_1$ to $g_2$ may be different from the effect from $g_2$ to $g_1$. For example, an improvement in system performance may not affect security, but increasing security will almost always adversely impact performance. we assume that the matrix $T$ is given prior to ADD; this assumption is reasonable as there is an effective map from any general scenario to the main quality attribute it tries to capture. We use this tradeoff matrix to define the interaction between two scenarios in $\NR$.
%\vspace*{-0.1cm}
\begin{definition}[Coalitional relevance]
For a coalition $D\subseteq \R$ and $r\in D$, the {\em coalitional relevance} of $r$ in $D$ is the total relevance from $r$ to all other requirements in $D$, i.e., $\rho(r,D)=\sum_{s\in D, s\neq r} \sigma(r,s).$% \vspace*{-0.2cm}$\rho(r,D)=\sum_{s\in D, s\neq r} \sigma(r,s).$$\vspace*{-0.4cm}
\end{definition}
%\vspace*{-0.1cm}
%For $r_1,r_2\in \NR$, the effect factor $\varepsilon(r_1,r_2,D)$ expresses the  effect of $r_1$ towards $r_2$ in the coalition $D$.
\begin{definition}[Effect factor]
For scenarios $r_1,r_2$ in the same coalition $D$, the {\em effect factor} from $r_1$ to $r_2$ expresses the  effect of $r_1$ towards $r_2$, i.e.,%\vspace*{-0.2cm}
% is the signed coalitional relevance of $r_1$, whose sign (positive or negative) is determined by the tradeoff matrix $T$. More specifically, we define
\[
    \varepsilon(r_1,r_2,D) = \begin{cases}
        -|\rho(r_1,D)| & \text{ if $T(g(r_1),g(r_2))=-1$} \\
        0 & \text{ if $T(g(r_1),g(r_2))=0$} \\
        \rho(r_1,D) & \text{ if $T(g(r_1),g(r_2))=1$}\vspace*{-0.2cm}
    \end{cases}
\]
\end{definition}
We are now ready to define the interaction between two scenarios $r_1,r_2\in \R$.
\begin{definition}[Interaction]
Let $r_1\!\neq\!r_2\!\in\!\R$ be requirements. The {\em interaction} between $r_1,r_2$ is simply the relevance $\sigma(r_1,r_2)$ if one of $r_1,r_2$ is functional; otherwise (both $r_1,r_2$ are non-functional), it is the sum of their effect factors, i.e., \vspace*{-0.2cm}% we define the {\em interaction function} $\nu$ as
\[
  \text{the {\em interaction} } \nu(r_1,r_2,D) \coloneqq \begin{cases}
        \sigma(r_1,r_2) & \text{ if $\{r_1,r_2\}\cap \F\neq \varnothing$}\\
        \varepsilon(r_1,r_2,D)+\varepsilon(r_2,r_1,D) & \text{ otherwise}%\vspace*{-0.1cm}
        \end{cases}
\]
The {\em coalition utility} $\nu(D)$ of any coalition $D\subseteq \R$ is defined as the sum of interactions among all pairs of requirements in the coalition, i.e.,\vspace*{-0.2cm}
\[
    \nu(D) = \sum_{r_1\neq r_2\in D} \nu(r_1,r_2,D)%\vspace*{-0.2cm}
\]
\end{definition}%\vspace*{-0.3cm}
\begin{definition}[Decomposition games (DG)]
Let $\A=(\FR, \NR, \C,\prec, \approx, \hookrightarrow)$ be an attribute primitive. The {\em DG} $G_\A$ is the coalition game $(\FR\cup \NR, \nu)$ where $\nu:2^{\FR\cup \NR}\to \mathbb{R}$ is the coalition utility function.
\end{definition}
\begin{figure}%\vspace*{-1cm}
\begin{center}\caption{\label{fig:relevance} \small (a) Weights on the edges are relevance (function $\sigma$) between requirements in Example~\ref{exp:relevance}; the diagram omits the negative weighted pairs. \ (b) The decomposition $\{S_1,S_2\}$ is a solution with $\nu(S_1)=2.5$, $\nu(S_2)=0.4$. The coalition $C$ has $\nu(C)=-1$}
\end{center} %\vspace*{-0.5cm}
 \includegraphics[width=12.0cm]{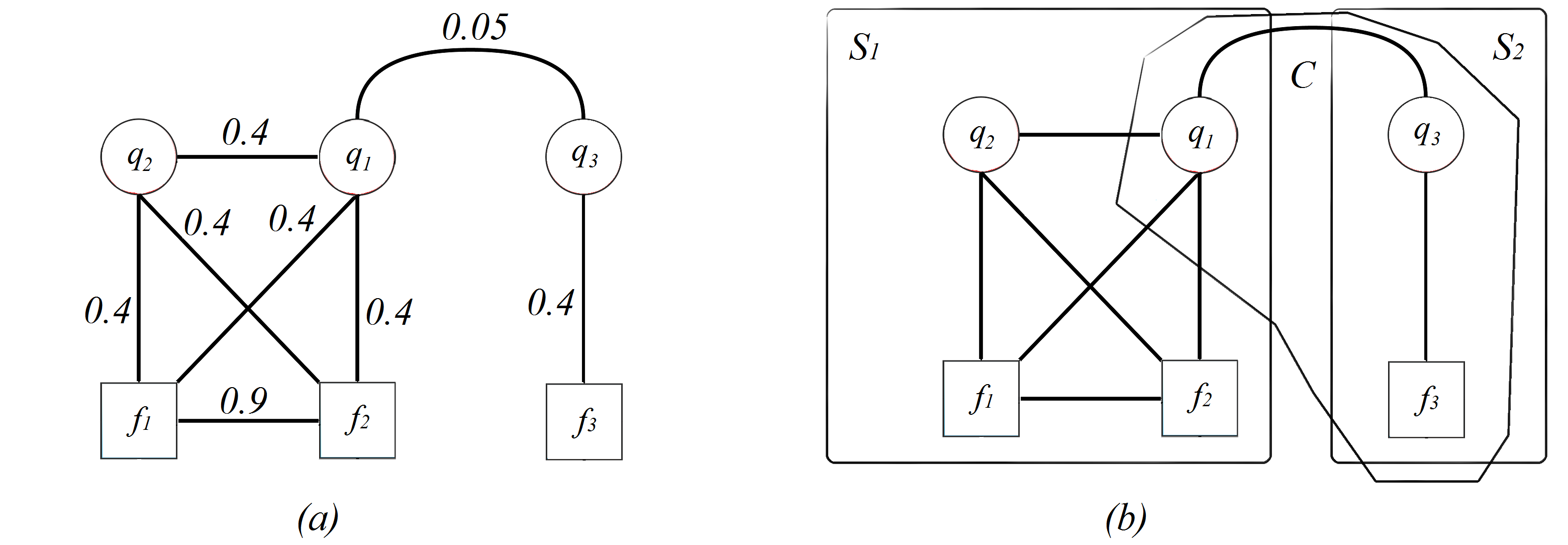}
 %\vspace*{-1cm}
\end{figure}

\begin{example}[Coalition Utility]\label{exp:utility}
Continue the setting in Example~\ref{exp:relevance}. Let the general scenarios be $g_1=\{q_1,q_2\}$ and $g_2=\{q_3\}$. We assume matrix $T$ specifies $T(g_1,g_2)=1$ and $T(g_2,g_1)=-1$. Consider the coalition $C=\{q_1,q_3,f_3\}$. We have:
\[
\rho(q_1,C)=0.05-0.5=-0.45; \text{ and }\rho(q_3,C)=0.4+0.05=0.45.
\]
So $\varepsilon (q_1,q_3,C)=-0.45\times 1=-0.45$ and $\varepsilon(q_3,q_1,C)=0.45\times (-1)=-0.45$.

Thus $\nu(q_1, q_3, C) =-0.45-0.45= -0.9$.
Therefore, $\nu(C) = \sigma(q_1,f_3)+\sigma(q_3,f_3)+(-0.9)=(-0.5)+0.4+(-0.9) = -1$ but $\nu(C \setminus \{q_1\}) = \nu(\{q_3,f_3\})=\sigma(q_3,f_3)=0.4$; See Fig.~\ref{fig:relevance}(b).

As it has turned out, despite the fact that matrix $T$ indicates $q_1$ will act positively to $q_3$, and that $q_1,q_3$ have a positive (0.05) relevance,  adding $q_1$ into the coalition of $\{q_3,f_3\}$ drastically decreases the coalition utility.
\end{example}
\subsection{Solution Concept}\label{subsec:solution}
We point out some major differences between decomposition and typical coalition games: Firstly, in  coalition game theory, one normally assumes the axioms of superadditivity ($\nu(D_1\cup D_2) \geq \nu(D_1)+\nu(D_2)$) and monotonicity ($D_1\subseteq D_2 \Rightarrow \nu(D_1)\leq \nu(D_2)$) which would obviously not hold for decomposition as players may counteract with each other, reducing their combined utility. Secondly, the typical solution concepts in coalition games (such as Pareto optimality,  and Shapely value) focus on distribution of payoffs to each individual player assuming a grand coalition consisting of all players. In decomposition such a grand coalition is normally not desirable and the focus is on the overall payoff of each coalition $D$, rather than the individual requirements.
%\end{itemize}
The above differences motivate us to consider a different solution concept of DG $G_\A$. At any instance of the game, the players form a decomposition $(D_1,D_2,\ldots,D_k)$. We assume %that the goal of each player is to optimise the overall payoff of its coalition. In this process,
that the players may perform  two collaborative strategies:
\begin{enumerate}
\item {\em Merge strategy}:  Two coalitions may choose to merge  if they would obtain a higher combined payoff.
\item {\em Bind strategy}: Players within the same coalition may form a sub-coalition if they would obtain a higher payoff.
\end{enumerate}
%\vspace*{-0.4cm}
\begin{example}[A Dilemma]\label{exp:dilemma}
We present an example demonstrating the dynamics of a DG $G_\A$. This example shows a real-world dilemma: As a coalition pursues higher utility through expansion (merging with others), %it may ``crumble'' from within; and when faced with two expansion strategies,
it may be better to choose a ``less-aggressive'' expansion strategy over the ``more-aggressive'' counterpart, even though the latter clearly brings a higher payoff. Assume the following  situation (which is clearly plausible in an attribute primitive):
\begin{itemize}
	\item $\R=\{d_1,d_2,d_3,d_4\}$ where $\NR=\{d_1,d_4\}$ and $d_1\not\approx d_4$.
    \item We set $\sigma(\{d_1, d_2\})=\sigma(\{d_1, d_3\})=\sigma(\{d_2,d_3\})=0.1$, and $\sigma(\{d_2, d_4\})=0.5$.
	\item The tradeoff matrix indicates $T(g(d_1), g(d_4))=0$, $T(g(d_4), g(d_1))=-1$.
	\item And, $d_1$ and $d_4$ are irrelevant, namely $\sigma(d_1, d_4)=\lambda = -0.7$. \vspace*{-0.2cm}
\end{itemize}
Suppose we start with the decomposition $\{S=\{d_1,d_2\}, \{d_3\}, \{d_4\}\}$.  Then $\nu(S)=\rho(d_1, d_2,S)=\nu(d_1, d_2,S)=0.1$.  Coalition $S$  has two merge strategies:
\begin{itemize}
	\item[(1)] For $S_1 = S\cup \{d_3\}$: $\nu(d_1, d_2,S_1) = \sigma(d_1, d_2) = 0.1$, $\nu(d_1, d_3,S_1)\!=\!\sigma(d_1, d_3)\!=\!0.1$, $\nu(d_2, d_3, S_1)\!=\!\sigma(d_2, d_3)\!=\!0.1$. Thus $\nu(S_1)\!=\!0.3$.
	\item[(2)] For $S_2 = S\cup \{d_4\}$: $\nu(d_1,d_4,S_2)\!=\!\varepsilon(d_4, d_1,S_2)\!=\!-0.7\!+\!0.5\!=\!-0.2$, $\nu(d_1, d_2,S_2)\!=\!\sigma(d_1, d_2)\!=\!0.1$ , $\nu(d_2, d_4,S_2)\!=\!\sigma(d_2, d_4)\!=\!0.5$. Hence $\nu(S_2)\!=\!0.1\!-\!0.2\!+\!0.5\!=\!0.4$
\end{itemize}
Merging with $\{d_4\}$ clearly results in a higher payoff for the combined coalition. However, if this merge happens, as $\nu\left(\{d_2, d_4\}\right)=0.5>\nu(S_2)=0.4$, $d_2$ and $d_4$ would choose to bind together, hence leaving $S_2$. This would be undesirable if $d_1$ is a  critical non-functional requirement for $d_2$.
\end{example}
Example~\ref{exp:dilemma} shows that a solution concept would be a decomposition where no ``expansion'' nor ``crumbling'' occur to any coalition. Formally, we define the following solution concepts:
\begin{definition}[Solution] Let $\vec{D}=(D_1,\ldots,D_k)$ be a decomposition of $\A$.
\begin{enumerate}
\item A coalition $D\subseteq \R$ is {\em cohesive} if for all $C\subseteq D$, $\nu(C) < \nu(D)$;
 $\vec{D}$ is {\em cohesive} if so is every $D_i$.
\item A coalition $D_i$ is {\em expansion-free} with respect to $\vec{D}$ if $\max\{\nu(D_i),\nu(D_j)\}>\nu(D_i\cup D_j)$;
$\vec{D}$ is {\em expansion-free} if so is every $D_i$.
\end{enumerate}
A {\em solution} of a DG is a decomposition that is both cohesive and expansion-free.
\end{definition}

\begin{example}[Solution]\label{exp:solution} Continue from Example~\ref{exp:utility}, the utilities for
\[
S_1 = \{q_1, q_2, f_1, f_2\} \text{ \qquad  and \qquad } S_2 = \{q_3, f_3\} \text{ \qquad are }\colon\]
%
%Computing the utility for each coalition to get:
\begin{itemize}
\item $S_1$: $\nu(q_1, q_2,S_1)= 0$,
$\nu(q_1, f_1,S_1)=\nu(q_1, f_2,S_1)=\nu(q_2, f_1,S_1)=0.4$, \\
$\nu(q_1, f_2,S_1)=0.4$, $\nu(f_1, f_2,S_1)=0.9$. Thus $ \nu(S_1)=0.4\times 4+0.9=2.5$
	
\item $S_2$: $w_(q_3, f_3,S_2)=0.4$. Thus $\nu(S_2)=0.4$
\end{itemize}
Both $S_1$ and $S_2$ are cohesive. Furthermore, we have
$\nu(q_1,q_3,\R)\!=\!0.75\!-\!1.05\!=\!-0.3$ and $\nu(q_2,q_3,\R)=0.2-1.05=-0.85$.
Thus $\nu(\R)\!=\!2.9\!-\!0.5\!\times\!6\!-\!0.85\!-\!0.3\!=\!-1.45$. Consequently, $\{S_1,S_2\}$ is also expansion-free, and is thus a solution of the game.
\end{example}

A solution of a DG $G_\A$ corresponds to a rational decomposition of the attribute primitive  $\A$. As shown by Thm.~\ref{thm:existence}, any attribute primitive admits a solution, and rather expectedly, a solution may not be unique.
\begin{theorem}[Solution Existence]\label{thm:existence}
There exists a solution in any DG $G_\A$.
\end{theorem}
\begin{proof}
We show existence of a solution by construction. Let $(D_1,D_2,\ldots,D_k)$ be a longest sequence such that for any $i=1,\ldots,k$, $D_i$ is a minimal coalition with maximal utility in $\R\setminus \{D_1,\ldots,D_{i-1}\}$ (i.e., $\forall D\subseteq \R\setminus \{D_1,\ldots,D_{i-1}\}: \nu(D_1)\geq \nu(D)$ and $\forall D\subseteq D_1: \nu(D_1)>\nu(D)$).

We claim that $\vec{D}=(D_1,\ldots,D_k)$ is a solution in $G_\A$. Indeed, for any $1\leq i\leq k$, any proper subset of $D_i$ would have payoff strictly smaller than $\nu(D_i)$ by minimality of $D_i$. Thus $\vec{D}$ is cohesive. Moreover, if $\nu(D_i\cup D_j)>\min\{\nu(D_i),\nu(D_j)\}$ for some $i\neq j$, then $D_{\min\{i,j\}}$ does not have maximal utility in $\R\setminus\{D_1,\ldots,D_{\min\{i,j\}-1}\}$. Hence $\vec{D}$ is expansion-free. \qed
\end{proof}

\begin{proposition}\label{prp:unique}
The solution of a DG may not be unique.
\end{proposition}
\begin{proof}
Let $\A=(\FR,\NR,\C, \prec,\approx,\hookrightarrow)$ be an attribute primitive where $\NR=\varnothing$ and $\FR=\{d_1,d_2,\ldots,d_6\}$. We may define $\C, \prec,\approx,\hookrightarrow$ in such a way that
\begin{itemize}
\item For all $\{i,j\}\subseteq \{1,2,3,4\}$ and $\{i,j\}\subseteq \{4,5,6\}$, $i\neq j \Rightarrow \nu(\{d_i,d_j\})=0.1$
\item For all $i\in \{1,2,3\}$, $j\in \{5,6\}$, $\nu(\{d_i,d_j\})=-0.1$
\end{itemize}
Consider $\vec{C}\!=\!\{C_1\!=\!\{d_1,d_2,d_3\},C_2\!=\!\{d_4,d_5,d_6\}\} \text{ and } \vec{D}\!=\!\{D_1\!=\!\{d_1,d_2,d_3,d_4\},$ $D_2\!=\!\{d_5,d_6\}\}$.
Note that $\nu(C_1)=0.3$ and $\nu(C_2) = 0.3$; $\vec{C}$ is cohesive and $\vec{C}$ is expansion-free as $\nu(\FR)=0.3=\nu(C_1)$.
Note also that  $\nu(D_1)=0.6$ and $\nu(D_2)=0.1$; $\vec{D}$ is cohesive and $\vec{D}$ is expansion-free as $\nu(D_1)>\nu(\FR)$\qed
\end{proof}

\section{Solving Decomposition Games} \label{sec:algorithm}
Based on our game model,  the operation $\mathsf{Decompose}(\A)$ in Procedure~\ref{prc:ADD} is reduced to  the following $\mathsf{DG}$ problem:

\qquad INPUT: An attribute primitive $\A=(\FR,\NR, \C,\prec,\approx,\hookrightarrow)$

\qquad OUTPUT: A solution $\vec{D}=(D_1,D_2,\ldots,D_k)$ of the game $G_\A$

\noindent Here, we measure computational complexity  with respect to the number of requirements in $\FR\cup \NR$. The proof of Theorem~\ref{thm:existence} already implies an algorithm for solving the $\mathsf{DG}$ problem: check all subsets of $\R$ to identify a minimal set with maximal utility; remove it from $\R$ and repeat. However, it is clear that this algorithm takes exponential time. We will demonstrate below that a polynomial-time algorithm for this problem is, unfortunately, unlikely to exist.

We consider the decision problem $\mathsf{DG\_D}$: {\em Given $\A$ and a number $w>0$,  is there a solution $\vec{D}$ of $G_\A$ in which the highest utility of a coalition reaches $w$?} Recall that the payoff function $\nu$ of $G_\A$ is defined assuming constants $\alpha,\beta,\gamma>0$ and $\lambda<0$. The theorem below holds assuming $\lambda<-\gamma$.
%Due to space limitation we only illustrate the main ideas of the proof.%; the full proof is in Appendix A.
%\vspace*{-0.2cm}
\begin{theorem}\label{thm:NPComplete}
The $\mathsf{DG\_D}$ problem is NP-hard.%\vspace*{-0.2cm}
\end{theorem}
\begin{proof}
The proof is via a reduction from the maximal clique problem, which is a well-known NP-hard problem.
Given an undirected graph $H=(V, E)$, we construct an attribute primitive $\A$ such that any cohesive coalition in $G_\A$ reveals a clique in $H$.  %Let $m$ be the smallest integer bigger than $\gamma/|\lambda|$.
Suppose $V=\{1,2,\ldots,n\}$. The requirements   of $\A$  consist of $n^2$ scenarios:
$\R=\NR\coloneqq \{a_{i,i'}\mid 1\leq\! i\!\leq\! n, 1\leq\! i'\!\leq\! n\}$. In particular, all requirements are non-functional.
We define an edge relation $E'$ on $\NR$ such that
\begin{enumerate}
\item $(i,j)\in E$ iff $(a_{i,i'},a_{j,j'})\in E'$ for some $1\leq i' \leq n$ and  $1\leq j'\leq n$
\item If $(a_{i,i'},a_{j,j'})\in E'$  then $(a_{i,i''},a_{j,j''})\notin E'$ for any $(i'',j'')\neq (i',j')$.
\item Any $a_{i,i'}$ is attached to at most one edge in $E'$.
\end{enumerate}
Note that such a relation $E'$ exists as any node $i\in V$ is only connected with at most $n-1$ other nodes in $H$.
Intuitively, a set of requirements $A_i=\{a_{i,1},\ldots,a_{i,n}\}$ serves as a ``meta-node'' and corresponds to the node $i$ in $H$.
In constructing $\A$, we may define the general scenarios  in such a way that
\begin{itemize}
\item $T(g(a_{i,j_1}), g(a_{i,j_2}))=0$ for any $1\leq i\leq n$ and $j_1\neq j_2$.
\item $T(g(a_{i_1,j_1}), g(a_{i_2,j_2}))=-1$ for any $(i_1,i_2)\notin E$.
\item $T(g(a_{i_1,j_1}), g(a_{i_2,j_2}))=1$ for any $(a_{i_1,j_1},a_{i_2,j_2})\in E'$
\item $T(g(a_{i_1,j_1}), g(a_{i_2,j_2}))=0$ for any $(i_1,i_2)\in E$ but $(a_{i_1,j_1},a_{i_2,j_2})\notin E'$
\end{itemize}
For every $1\leq i\leq n$ and $1\leq j<j'\leq n$, put in a constraint $c_{i}(j,j')=\{a_{i,j}, a_{i,j'}\}$.   Thus the relevance between $a_{i,j}$ and $a_{i,j'}$ is \vspace*{-0.1cm}
\[
    \sigma(a_{i,j},a_{i,j'})=\frac{|c(a_{i,j})\cap c(a_{i,j'})|}{|c(a_{i,j})\cup c(a_{i,j'})|}=\frac{\gamma}{2(n-1)}
\]%\vspace*{-0.1cm}
Furthermore if $i\neq i'$, then for any $j,j'$ we set $\sigma(a_{i,j},a_{i,j'})=\lambda$.
Suppose $U=\{i_1,\ldots,i_\ell\}$ induces a complete subgraph of $H$. We define the {\em meta-clique coalition} of $U$ as
\[
D_U=\bigcup_{1\leq j\leq \ell} A_{i_j}
\]
By the above definition, for any $1\!\leq\!s\!<\!t\!\leq\!\ell$, take $j,j'$ such that $(a_{i_s,j},a_{i_t,j'})\in E'$. %We have\vspace*{-0.1cm}
\begin{align*}
w(i_s,i_{t},D_U) & = \varepsilon(a_{i_s,j},D_U) + \varepsilon(a_{i_t,j'},D_U) \\
& = \rho(a_{i_s,j},D_U) + \rho(a_{i_t,j'},D_U)\\
& = (n-1)\times \frac{\gamma}{2(n-1)} + (n-1)\times \frac{\gamma}{2(n-1)}  =\gamma%\\
%& = \gamma
\end{align*}
Thus $\nu(D_U)=\frac{n(n-1)\gamma}{2}$. Taking out any element from $D_U$ results in a strict decrease in utility, and hence $D_U$ is cohesive.

Now take any coalition $D\subseteq \R$ that contains two requirements $a_{i,i'}$, $a_{j,j'}$ such that $(i,j)\notin E$.
Let $s=|A_i\cap D|$ and $t=|A_j\cap D|$.  Note also that
$\sigma(a_{j,j'},a_{i,i''}) = \lambda$ for any $a_{i,i''}\in A_i\cap D$. Therefore we have\vspace*{-0.1cm}
\[
  \nu(D)- \nu(D\setminus \{a_{j,j'}\})  \leq \gamma+2 w(a_{j,j'},a_{i,i'},D)\times s \leq  \gamma+2\lambda +\gamma = 2(\lambda+\gamma)<0
\]\vspace*{-0.1cm}
The last inequality above is by assumption that $\lambda<-\gamma$. Thus $D$ is not cohesive.

By the above argument, a coalition $D\subseteq \R$ is cohesive in $G_\A$ iff $D$ is the meta-clique coalition $D_U$ for some clique $U$ in $H$. Furthermore, a decomposition $\vec{D}=(D_1,D_2,\ldots,D_k)$ is a solution in $G_\A$ iff $V$ can be partitioned into sets $U_1,\ldots,U_k$ where each $U_i$ is a clique, and $D_i=D_{U_i}$ for all $1\!\leq\!i\!\leq\!k$.
 In particular, $H$ has a clique with  $\ell$ nodes if and only if $G_\A$ has a solution that contains a coalition whose utility reaches  $\frac{\ell(\ell-1)\gamma}{2}$. This finishes the reduction.\qed
\end{proof}
Theorem~\ref{thm:NPComplete} shows that, in a sense, identifying a ``best'' solutions in a DG $G_\A$ is hard. The main difficulty comes from the fact that one would examine all subsets of players to find an optimal cohesive coalition.  This calls for a relaxed notion of a solution that is computationally feasible. To this end we introduce the notion of {\em $k$-cohesive coalitions}. Fix $k\in \N$ and enforce this rule: Binding can only take place on $k$ or less players. That is, a coalition $C$ is {\em $k$-cohesive} whenever $\nu(C)$ is greater than the utility of any subsets with at most $k$ players.\vspace*{-0.1cm}

\begin{definition} Fix $k\in \N$. In a DG $G_\A =  (\FR \cup \NR, \nu)$, we say a coalition $D \subset \FR \cup \NR$ is {\em $k$-cohesive} if $\nu(D') < \nu(D)$ for all $D' \subset D$ with $|D'| \leq k$. An decomposition $\vec{D}$ is {\em $k$-cohesive} if every coalition in $\vec{D}$ is $k$-cohesive; if $\vec{D}$ is also expansion-free, then it is a {\em $k$-cohesive solution} of the game $G_\A$.
\end{definition}

\paragraph*{\textit{Remark.}} In a sense, the value $k$ in the above definition indicates a level of {\em expected cohesion} in the decomposition process. A higher value of $k$ implies less restricted binding within any coalition, which results in higher ``sensitivity'' of design elements to conflicts. In a software tool which performs ADD based on DG, the level $k$ may be used as an additional parameter.

%\vspace*{-0.5cm}
\begin{algorithm}
\caption{\label{proc:kcombine}$\mathsf{DGame}(\A,k)$}
\begin{algorithmic}[1]
   \INPUT Attribute primitive $\A$, $k>0$
   \OUTPUT Attribute Decomposition $\vec{D}$
\State $\vec{D} \leftarrow \mathsf{Cohesives}(\A, k)$
\State $\mathsf{Combine}\leftarrow \true$
\While{$\mathsf{Combine}$}
    \State  $\mathsf{Combine}\leftarrow \false$
    \For{$(D,D') \in \vec{D}^2$, $D\neq D'$}
		 \If{$\nu(D' \cup D) > \nu(D)$ and $\nu(D' \cup D) > \nu(D')$}
			\State $D \leftarrow D' \cup D$ and remove $D'$ from $\vec{D}$
		 	\State $\mathsf{Combine}\leftarrow \true$
		 	%\State {\bf break}
		 \EndIf
    \EndFor	
\EndWhile
\State {\bf return} $\vec{D}$
\end{algorithmic}
\end{algorithm}
\vspace*{-1cm}
\begin{algorithm}
\caption{\label{prc:kCohesive}$\mathsf{Cohesive}(\A, k)$}
\begin{algorithmic}[1]
   \INPUT Attribute primitive $\A$, $k>0$
   \OUTPUT Attribute Decomposition $\vec{D}$
\State $\vec{D} \leftarrow [\ ]$, $R \leftarrow \FR \cup \NR$	
\While{$|R| > 0$}
%// Compute the k-cohesive coalition with  the largest payoff in $D$
	\State $S \leftarrow  \mathsf{max}(R,k)$ // compute a maximally $k$-cohesive coalition
	\State $R \leftarrow R\setminus S$
	\State $\vec{D} \leftarrow \mathtt{[}\vec{D},S\mathtt{]}$
\EndWhile
\State {\bf return} $\vec{D}$
\end{algorithmic}
\end{algorithm}

Let $R$ be a set of requirements. A coalition $D$ is called {\em maximally $k$-cohesive} in $R$ if $|D|\leq k$, $D$ is $k$-cohesive and $\nu(D)\geq \nu(D')$ for any $D'\subseteq R$. Suppose the operation $\mathsf{max}(R,k)$ computes a maximally $k$-cohesive set in $R$.
The algorithm $\mathsf{DGame}(\A,k)$ (Proc.~\ref{proc:kcombine}), which uses $\mathsf{Cohesive}(\A, k)$ (Proc.~\ref{prc:kCohesive}) as a subroutine, computes a $k$-cohesive solution of $G_\A$.
Note that the $\mathsf{Cohesive}(\A,k)$ operation maintains a list $\vec{D}$, which when returned, denotes a decomposition. Note also that the returned $\vec{D}=(D_1,\ldots,D_m)$ satisfies the following condition: \vspace*{-0.2cm}
\[
\forall 1\leq i\leq m: D_i \text{ is maximally $k$-cohesive in $D_{i}\cup \cdots \cup D_m$}
\]
We call this $\vec{D}$ a {\em maximally $k$-cohesive decomposition}.

\begin{lemma}\label{lem:ksolution}
Suppose $\vec{D}$ is a maximally $k$-cohesive decomposition. Take any $1\!\leq\!i\!<\!j\!\leq\!n$.
If $\nu(D_i\cup D_j)>\max\{\nu(D_i),\nu(D_j)\}$ then $D_i\cup D_j$ is $k$-cohesive.
% If for all $D_i, D_j$: $ \nu(D_i \cup D_j) < \nu(D_i)$ or $ \nu(D_i \cup D_j) < \nu(D_i)$, then the sequence is $k$-solution.
\end{lemma}

\begin{proof} Let $S_i=\bigcup_{i\leq j\leq m} D_j$ for any $i=1,\ldots, m$.
Suppose $\nu(D_i\cup D_j)>\max\{\nu(D_i),\nu(D_j)\}$ for $1\leq i<j\leq m$. By assumption $D_i$ is maximally $k$-cohesive in $S_i$. For any finite set $U\subseteq D_i\cup D_j\subseteq S_i$ such that $|U|\leq k$, we have $\nu(U)\leq \nu(D_i) < \nu(D_i\cup D_j)$. Hence $D_i\cup D_j$ is also $k$-cohesive. \qed
\end{proof}

\begin{theorem}\label{thm:algo} Given an attribute primitive $\A$,
the $\mathsf{DGame}(\A,k)$  algorithm  computes a $k$-cohesive solution of the decomposition game $G_\A$ in time $O(n^{k})$, where $n$ is the number of requirements in $\A$.
\end{theorem}

\begin{proof}
The  $\mathsf{DGame}(\A,k)$ algorithm calls $\mathsf{Cohesive}(\A,k)$ to produce a maximally $k$-cohesive decomposition $\vec{D}$, and then performs several iterations to ``combine'' the coalitions in $\vec{D}$. By Lemma~\ref{lem:ksolution}, the decomposition $\vec{D}$ after each iteration is $k$-cohesive. There is a point when for all $D,D'\in \vec{D}$ we have $\nu(D\cup D') \leq \max\{\nu(D),\nu(D')\}$. At this moment,  the {\bf while}-loop will terminate and $\vec{D}$ is expansion-free.
The time complexity is justified as there are $O(n^k)$ subsets of $\FR\cup \NR$ with size $\leq k$. Thus computing a maximally $k$-cohesive decomposition takes time $O(n^{k})$. \qed
\end{proof}

\section{Case Study: Cafeteria Ordering System} \label{sec:case}
%\subsection{Overview}
\begin{wrapfigure}{r}{5.5cm}\vspace*{-1.1cm}
 \caption{\label{fig:cafe} \small Interactions between requirements in the COS \cite{WKJ2013}. Blue edges indicate positive interactions and red edges indicate negative interactions.}
 \includegraphics[width=6cm]{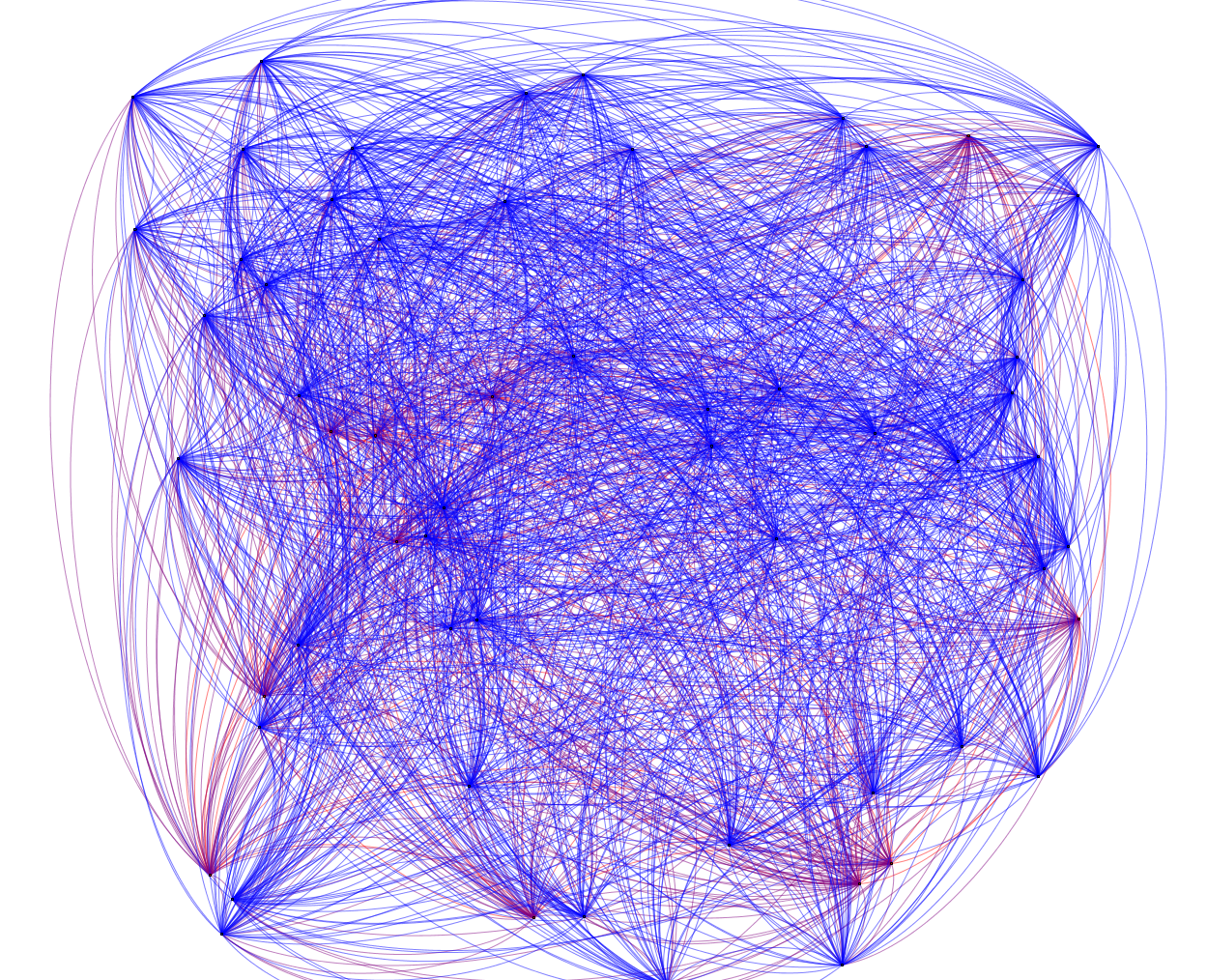} \vspace*{-1cm}
\end{wrapfigure}
To demonstrate applicability of our game model in real-world, we build a DG for a cafeteria ordering system (COS). A COS permits employees of a company to order meals from the company cafeteria online and is a module of a larger cafeteria management system. The requirements of the project have been produced through a systematic requirement engineering process and is well-documented (See full details from \cite[Appendix C]{WKJ2013}).
Since COS is a subsystem within a larger system, the requirements also incorporate interfaces with other subsystems of the overall system.  The initial attribute primitive has 60 requirements with  $|\NR|=11$, $|\FR|=49$ and 7 design constraints.
Non-functional requirements conflict with each other, e.g., the general scenario $\mathsf{USE}$ conflicts with the general scenario $\mathsf{PER}$. Also the requirements exhibit some complex relationships, e.g. $\mathsf{SEC1} \hookrightarrow\mathsf{Order.Pay.Deduct}$.

We demonstrate the complicated interactions among requirements using a complete graph where nodes are all requirements in $\R=\NR\cup \FR$; see  Fig.~\ref{fig:cafe}. The edges are in two colours: $(r_1,r_2)$ gets blue if $w(r_1,r_2, \R)>0$ and red if $w(r_1,r_2,\R)<0$.
%For completeness we describe requirements and constraints below:
%Here we list constraints and requirements that are relevant to the ordering system.
%The complete requirement specification can be found in Appendix C of the book \cite{WKJ2013}.
{\em (For completeness, we include descriptions of constraints and requirements in the APPENDIX.)}

We run the $\mathsf{DGame}(\A,k)$ algorithm  to identify a $k$-cohesive solution for different levels $k$ of expected cohesion. In order to clearly identify sub-components, we give a higher penalty $\lambda$ between conflicting requirements: $\alpha=0.4$, $\beta=0.3$, $\gamma=0.3$, $\lambda=-1.3$. We choose  $k\in \{1,\ldots,7\}$. As argued above, setting a higher value of $k$ should in principle improve the quality of the output decomposition, although this also means a longer computation time. We implement our algorithm using Java  on a laptop with Intel Core i7-3630QM CPU 2.4GHz 8.0GB RAM. The running time for different values of $k$ is: 503 milliseconds for $k=3$ and approximately $1140$ seconds for $k=6$.

\begin{table}\caption{\footnotesize Resulting $3$- and $6$-cohesive solutions, ordered by payoff values.}
{\scriptsize \centering
\begin{tabular}{|c|c|c|}
\hline           &  $3$-Cohesive Solution & $6$-Cohesive Solution \\ \hline
Coalition 0& \begin{tabular}{@{}c@{}} AVL1 ROB1 SAF1 SEC(1,2,4) USE(1,2)\\ Order.Confirm Order.Menu.Data \\Order.Deliver.(Select,Location)\\  Order.Pay Order.Place Order.Retrieve \\  Order.Units.Multiple UI2 UI3  \end{tabular} &  \begin{tabular}{@{}c@{}} AVL1 ROB1 SAF1 SEC(1,2,4) PER(1,2,3)\\ USE(1,2) Order.Confirm Order.Deliver \\Order.Deliver.(Select,Location) \\  Order.Menu.Date Order.Pay \\ Order.Retrieve  Order.Place Order.Units \\Order.Units.Multiple  UI2 UI3  \end{tabular}  \\\hline
Coalition 1& \begin{tabular}{@{}c@{}}PER(1,2,3) Order.Units.TooMany \\Order.Deliver.(Times,Notimes)\\
Order.Place.(Cutoff,Data,Register,No)\\ Order.Pay.(OK,NG) Order.Done.Failure \\Order.Confirm.(Prompt,Response,More)\end{tabular}& \begin{tabular}{@{}c@{}}Order.pay.(Deliver,Pickup,Deduct)\\Order.Done.Patron SI2.2 SI2.3 \end{tabular} \\\hline
Coalition 2& \begin{tabular}{@{}c@{}}Order.pay.(Deliver,Pickup,Deduct)\\Order.Done.Patron SI2.2 SI2.3 \end{tabular} & \begin{tabular}{@{}c@{}} Order.Units.TooMany \\Order.Deliver.(Times,Notimes)\\
Order.Place.(Cutoff,Data,Register,No)\\ Order.Pay.(OK,NG) Order.Done.Failure \\Order.Confirm.(Prompt,Response,More)\end{tabular} \\ \hline
Coalition 3& \begin{tabular}{@{}c@{}} Order.Menu Order.Unit Order.Done\\Order.Done.(Menu,Times,Cafeteria) \\Order.Done.(Store,Inventory)\\Order.Deliver Order.Menu.Available\\Order.Confirm.Display\\Order.Pay.Method  SI1.3 SI2.5 CI2\end{tabular} & \begin{tabular}{@{}c@{}}  Order.Done\\Order.Done.(Menu,Times,Cafeteria) \\Order.Done.(Store,Inventory)\\  SI1.3 SI2.1 SI2.4 SI2.5 CI1 CI2\end{tabular} \\ \hline
Coalition 4 & SI1.1 SI1.2 & Order.Menu.Available SI1.1 SI1.2 \\\hline
\end{tabular}
}
\end{table}
%The detailed solutions are left out due to space limitation.%See Appendix~D for the detailed solutions.

\paragraph*{\bf Cohesion level $k= 3$.} The $3$-cohesive solution consists of 5 coalitions. An examination at the requirements in each coalition reveals: {\em Coalition 0} relates to usability and ensures availability of user interactions; it apparently corresponds to a user interface module. {\em Coalition 1} is performance-oriented and is separated from the usability requirements; it thus corresponds to a back-end module that handles all the internal operations. {\em Coalition 2} deals with the payroll system outside COS and defines a controlling interface from COS to payroll. {\em Coalition 3} consists of several functional requirements that control life cycle of the COS. {\em Coalition 4} is an interface to access the inventory system outside COS.

 It is clear that this solution separates the control, user inputs and computation modules, and fits the MVC (Model-View-Controller) architectural pattern. In addition, there is a design constraint that requires  the use of Java and Oracle database engine. So, we instantiate the design elements as in Fig.~\ref{fig:difficulty3}.

\begin{figure}
\begin{center}
\caption{\label{fig:difficulty3} {\footnotesize The $3$-cohesive solution. \emph{Coalition 0}: Java Spring framework uses server page as user interface and provides a powerful encryption infrastructure (Spring Crypto Module). Server page is suitable for implementing interactive user interface. {\em Coalition 1}: Enterprise Java Bean (EJB) is a middleware (residing in the application server) used to communicate between different components. It provides rich features for processing HTTP requests. {\em Coalition 2}: The COS uses a package solution from corresponding payroll system. {\em Coalition 3}: A servlet is a controller in Java application server which separates business logic from control.  {\em Coalition 4}:   A web service interface  outside COS.}}
\includegraphics[width=11cm]{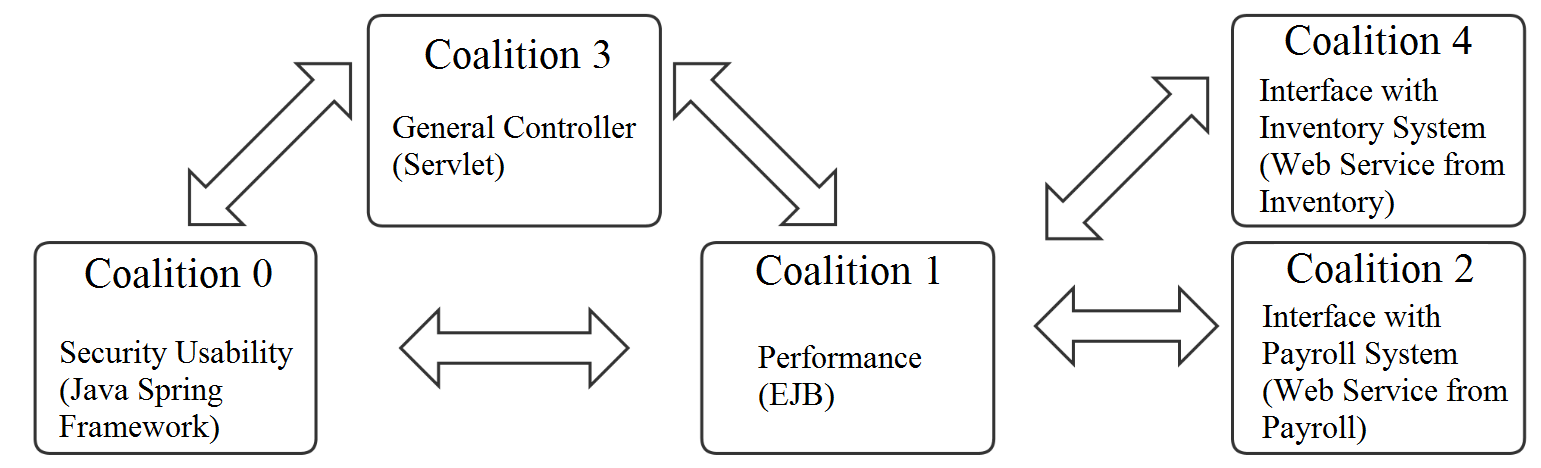}
\end{center}
\end{figure}
\paragraph*{\bf Cohesion level $k=6$.} The $6$-cohesive solution also contains five coalitions, with a similar structure as the $3$-cohesive counterpart. There are, nevertheless, several important differences:
%As expected, we achieve a better resolution of conflicts in the $6$-cohesive solution than the previous case. This can be told in several aspects.
Firstly, the performance ($\mathsf{PER}$) scenarios now belong to coalitions 0. This means that some performance-related computation is moved to the front-end. This is reasonable as this lightens the computation load of the back-end and thus improving performance and availability.
Secondly, the functional requirement $\mathsf{Order.Menu.Available}$ is moved to coalition 4, which is the interface between COS and the inventory system. This requirement specifies that the menu should only display those food items that are available in inventory. %Thus we argue this solution is more reasonable because this functionality can be naturally performed in the interface between COS  and inventory.

%One design constraint of the system is that COS should use HTTP as the only communication protocol.
%Using HTTP for communication is heavy but we also want to make improvement for COS. So, we can take suggestions from our new solution. This solution essentially sets the front-end of COS as a module.
Instead of server page, we use scripting to reduce the server's computation load. This can be achieved by changing the front-end to a JavaScript oriented designs. The main difficulty lies in that we need to put extra effort when using JavaScript to communicate with web server (such as AJAX) in order to ensure usability, performance and security. We instantiate design elements as in Fig.~\ref{fig:difficulty6}.

\begin{figure}
\begin{center}\caption{\label{fig:difficulty6} \footnotesize The $6$-cohesive solution. {\em Coalition 0} uses JavasScript as a front end for user interface. It also takes some computation for sever in order to achieve better performance. {\em Coalition 1} is an interface for accessing the payroll system. {\em Coalition 2} ensures the business logic in COS. {\em Coalition 3} coordinates input from front end (coalition 0) to back end (coalition 1). {\em Coalition 4} is an interface for accessing the inventory system.}
\includegraphics[width=11cm]{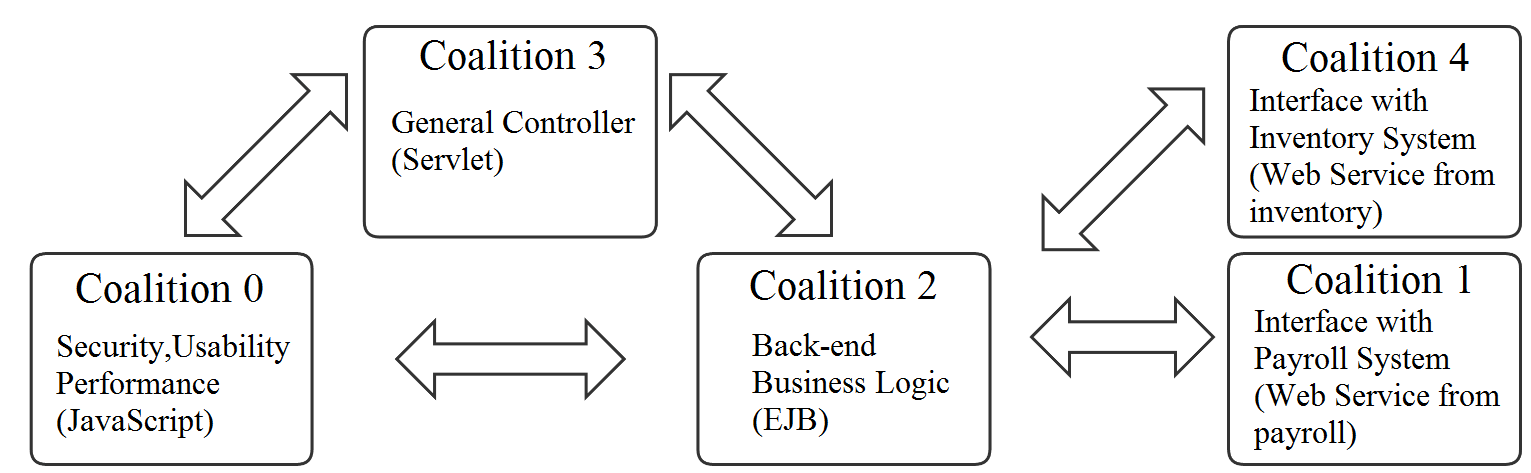}
\end{center}
\end{figure}

\section{Related Work}\label{sec:related}
Bass, Klein and Bachman introduce ADD as a general framework for developing conceptual architecture in \cite{BLMF2002}. They argue that non-functional requirements should drive decision making throughout the entire design process. Furthermore, as non-functional requirements often provide a high-level view of a software system, ADD should follow an iterative decomposition process. They further improve their method in \cite{WR2006} by clarifying how ADD is carried out in a real life project. The technique of evaluating tradeoff between quality attributes in these mentioned works is largely empirical-based.
Kazman et. al. first investigate tradeoff between quality attributes in \cite{KR1998}. They collect and analyse design elements that affect  multiple quality attributes. This approach aims to mitigate risks residing in a software architecture and refine design through this process. The study is further investigated in \cite{ZL2005} which gives a quantitative tradeoff analysis for non-functional requirements, and prioritises non-functional requirements during the ADD process. In \cite{AADHMH2011}, a different approach is provided to elicit non-functional requirements. The work follows the paradigm in \cite{WR2006} but generates more detailed and concrete designs; it computes tradeoff between non-functional requirements based on relationships between non-functional and functional requirements. There are also other algorithmic methods for software architecture design from the perspective of component decomposition. For example, in \cite{LXZ2007}, the authors propose a hierarchical clustering algorithm to decompose functional requirements and non-functional requirements. They label each component with a set of attributes and identify similarities between components based on their common attributes. Hence their approach does not put emphasis on the enhancement and conflict between attributes.

\section{Conclusion and Future work} \label{sec:conclusion}
%Analysing tradeoff among requirements in software architecture design has been an important question \cite{AADHMH2011,KR1998,LXZ2007}.
The use of computational games in software architecture design is a novel technique aimed to contribute to this line of direction. We proposed a game-based approach that, not only builds on established software architecture research (ADD), but is also shown --- through a case study --- to provide reasonable design guidelines to a real world application. We suggest that this framework would be useful in the following:
\begin{itemize}
\item Designing a software system that involves a large number of functionalities and quality attributes, which will result in a complicated architecture design
\item Designing a software system that hinges on the satisfaction of certain core quality attributes
\item Evaluating and analysing the rationale of an architecture design in a formal way; identifying potential risks with a design.
\end{itemize}
It is noted that the framework described here assumes the completion of  requirement analysis. In real life requirements are usually identified as the software is implemented (e.g. the agile software development methodology).
%While our game would still give design guidelines with partially completed requirements,
It would thus be interesting to develop a dynamic version of the game model, which supports architectural design using incremental refinements. Another future work is to develop a mechanism which maps coalitions generated by the algorithm to appropriate attribute primitives. This would then lead to a full automation of the ADD process linking requirements to conceptual architecture designs. %\vspace*{-0.3cm}

%{\Large Program output of the $3$-cohesive solution and te $6$-cohesive solution for the COS case study}

%\bigskip

%\noindent

%\newpage

\bigskip

\bigskip

\bigskip

%\noindent\makebox[\linewidth]{\rule{\paperwidth}{0.4pt}}

\bigskip

\bigskip

\bigskip

\newpage

\noindent {\Large {\bf APPENDIX:} Constraints and Requirements for COS Case Study}

\bigskip

{\noindent
\noindent {\bf Design Constraints}
\begin{itemize}
	\item CO-2: The system shall use the current corporate standard Oracle database engine.
    \item CO-3: All HTML code shall conform to the HTML 5.0 standard.
    \item BR-2: Deliveries must be completed between 11:00am and 2:00pm local time.
    \item BR-3: All meals in a single order must be delivered to the same location.
	\item BR-8: Meals must be ordered within 14 calendar days of the meal date.
	\item BR-11: If an order is to be delivered, the patron must pay by payroll deduction.
    \item BR-33: 256-bit encryption or network transmissions that involve financial information% or personally identifiable information require
\end{itemize}

\noindent {\bf Functional Requirements}
\begin{itemize}
\item {\em Order.Place}: Placing a meal order
\begin{itemize}
\item[] {\em .Register}: Confirm that the Patron is registered for payroll reduction.
\begin{itemize}
\item[] {\em .No}: If the patron is not registered for payroll deduction, the COS shall give the Patron options to register now and continue placing an order%, to place an order for pickup in the cafeteria (but not for delivery), or to exit.
\end{itemize}
\item[] {\em .Date}: The COS shall prompt the Patron for the meal date (See BR-8)
\begin{itemize}
    \item[] {\em .Cutoff}: If the meal date is today and is after the cutoff time, inform the Patron that it's too late %to place an order for today.
        The Patron can either change the meal date or cancel the order.
\end{itemize}
\end{itemize}
\item  {\em Order.Deliver}: Delivery or pickup
\begin{itemize}
    \item[] {\em .Select}: The Patron specifies whether the order is to be picked up or delivered
    \item[] {\em .Location}: If the order is to be delivered and there are still available delivery times for the meal date, the Patron shall provide a valid delivery location.
    \item[] {\em .Notimes}: Notify the Patron if there are no available delivery times. The Patron shall either cancel or pick up the order in the cafeteria.
    \item[] {\em .Times}: Display the remaining available delivery times for the meal date, allowing the Patron to request one of the times shown
\end{itemize}
\item  {\em Order.Menu}: Viewing a menu
\begin{itemize}
\item[] {\em .Date}: Display a menu for the date that the Patron specified.
\item[] {\em .Available}: The menu for the specified date shall display only those food items for which at least one unit is available in the cafeteria's inventory and which can be delivered
\end{itemize}
\item  {\em Order.Units}: Ordering multiple meals and multiple food items
\begin{itemize}
    \item[] {\em .Multiple}: Permit the user to order multiple identical meals %, up to the fewest available units of any menu item in the order
    \item[] {\em .TooMany}: If the Patron orders more units of a menu item than are presently in the cafeteria's inventory, inform  the maximum number of units that can order.
\end{itemize}
\item  {\em Order.Confirm}: Confirming an order
\begin{itemize}
    \item[] {\em .Display}: When the Patron indicates  no wish to order any more food items, display the ordered itmes,  prices, and the payment amount
    \item[] {\em .Prompt}: Prompt the Patron to confirm the meal order.
    \item[] {\em .Response}: The Patron can confirm, edit, or cancel the order.
    \item[] {\em .More}: Let the Patron order additional meals for the same or for a different date. %BR-3 and BR-4 pertain to multiple meals in a single order
\end{itemize}
\item  {\em Order.Pay}: Meal order payment
\begin{itemize}
\item[] {\em .Method}: When the Patron indicates that he is done placing orders, the COS shall ask the user to select a payment method
\item[] {\em .Deliver}: Se BR-11
\item[] {\em .Pickup}: If the meal is to be picked up in the cafeteria, the Patron shall choose to pay by payroll deduction or by cash at the time of pickup
\item[] {\em .Deduct}: If Patron selects payroll deduction, issue a payment request to Payroll
\item[] {\em .OK}: If the payment request is accepted,  display confirmation a message.
\item[] {\em .NG}: If the payment request is rejected, display the reason for the rejections. %The Patron shall either cancel the order, or change the payment method to cash and request to pick up the order at the cafeteria.
\end{itemize}
\item {\em Order.Done}: Finishing the process after the Patron confirms the order
\begin{itemize}
\item[] {\em .Store}: Assign the next available meal order number to the meal and store the meal order
\item[] {\em .Inventory}: Send a message to the inventory system with the number of  units
\item[] {\em .Menu}: Update the menu for the current order's order date to reflect any items that are now out of stock in the cafeteria inventory
\item[] {\em .Times}: Updates the remaining available delivery times for the date of this order
\item[] {\em .Patron}: Send email message to the Patron with the meal order and payment info
\item[] {\em .Cafeteria}: Send an email message to the Cafeteria Staff with the meal order information
\item[] {\em .Failure}: If any step of Order.Done fails, roll back the transaction and notify the user
\end{itemize}

\item {\bf User  Interfaces}
\begin{itemize}
	\item {\em UI2}: Provide a help link from each displayed webpage to explain how to use that page.
    \item {\em UI3}: The webpages shall permit complete navigation and food item selection
\end{itemize}
\item {\bf Software Interfaces}
\begin{itemize}
	\item {\em SI1.1}: Transmit the quantities of food items ordered to the Cafeteria Inventory
System through a programmatic interface.
	\item {\em SI1.2}: Poll the Inventory System to determine whether a requested item is available.
    \item {\em SI1.3}: When the Cafeteria Inventory System notifies the COS that a specific food item is not available, the COS shall remove that food item from the menu for the current date.
    \item {\em CI1}: Send an email or text message to the Patron to
confirm order acceptance%, price, and delivery instructions.
    \item {\em CI2}: Send an email or text message to the Patron to
report any problems %with a meal order or delivery.
\end{itemize}
\end{itemize}

\noindent {\bf Non-functional Requirement} We have 6 general scenarios: \textit{USE, PER, SEC, SAF, AVL, ROB}. Each of them associates multiple scenarios. %For example, general scenario \textit{USE} has an instance scenario \textit{USE-1}: ``The COS shall allow a Patron to retrieve the previous meal ordered with a single interaction.''% The non-functional requirements are prioritized and listed in order of importance.

\begin{itemize}
	\item {\em USE1}: Allow a Patron to retrieve the previous meal ordered with a single interaction.
    \item {\em USE2}: 95\% of new users shall be able to order a meal without errors on their first try.
    \item {\em PER1}: Accommodate 400 users and up to 100 concurrent users during the peak usage time, with an estimated average session duration of 8 minutes.
	\item {\em PER2}: 95\% of webpages generated shall download completely within 4 seconds from the time the user requests the page over a 20 Mbps or faster Internet connection.
	\item {\em PER3}: Display confirmation messages to users within an average of 3 seconds and a maximum of 6 seconds after the user submits information to the system.
	\item {\em SEC1}: All network transactions that involve financial information or personally identifiable info shall be encrypted per BR-33.
    \item {\em SEC2}: Users shall be required to log on for all operations except viewing a menu.
    %\item SEC-3: Only authorized Menu Managers shall be permitted to work with menus, per BR-24.
    \item {\em SEC4}: The system shall permit Patrons to view only orders that they placed.
    \item {\em SAF1}: The user shall be able to see all ingredients in any items, with allergic reactions.
    \item {\em AVL1}: The COS shall be available at least $98\%$ of the time between 5am and midnight and at least $90\%$ of the time, excluding scheduled maintenance windows.
    %\item SAF-1: The user shall be able to see a list of all ingredients in any menu items, with ingredients highlighted that are known to cause allergic reactions in more than 0.5 percent of the North American population.
    \item {\em ROB1}: If the connection between the user and the COS is broken prior to a new order being either confirmed or terminated, the COS shall enable the user to recover an incomplete order and continue working on it.
\end{itemize}
}

\end{document}